\documentclass{article}

\usepackage[square,sort,comma,numbers]{natbib}
\usepackage{relsize}
\usepackage[svgnames,dvipsnames]{xcolor}
\usepackage{amssymb}
\usepackage[T1]{fontenc}
\usepackage{mathtools}
\usepackage{breqn}
\newcommand{\sfrac}[2]{\frac{#1}{#2}}
\usepackage{braket}
\usepackage{physics}
\usepackage{tikz}
\usepackage{xparse}
\usepackage{floatrow}
\usepackage{varwidth}
\usepackage[inline]{enumitem}
\usepackage{bm}
\usepackage[unicode]{hyperref}
\usepackage[margin=1.25in,headheight=13.6pt]{geometry}
\usepackage{tasks}
\usepackage[utf8]{inputenc}

\usepackage{amsthm}

\usepackage[noabbrev,capitalize,nameinlink]{cleveref}
\usepackage{changepage}
\usepackage[normalem]{ulem}
\usepackage{booktabs}
\usepackage{array}
\usepackage{graphicx}
\usepackage{soul}
\usepackage[shortcuts]{extdash}
\usepackage{appendix}
\usepackage{pbox}
\usepackage{sectsty}
\usepackage{microtype}
\usepackage{flafter}

\hypersetup{%
  colorlinks,
  bookmarksnumbered,
  pdfencoding=unicode,
  allcolors=DarkSlateGrey
}

\newcolumntype{L}[1]{>{\raggedright\let\newline\\\arraybackslash\hspace{0pt}}m{#1}}
\newcolumntype{C}[1]{>{\centering\let\newline\\\arraybackslash\hspace{0pt}}m{#1}}
\newcolumntype{M}[1]{>{\centering\let\newline\\\arraybackslash\hspace{0pt}$}m{#1}<{$}}
\newcolumntype{R}[1]{>{\raggedleft\let\newline\\\arraybackslash\hspace{0pt}}m{#1}}

\usepackage{setspace}
\usepackage{eqparbox, etoolbox}
\newlist{rdescription}{description}{1}
\AtBeginEnvironment{rdescription}{%
\setlist[rdescription]{leftmargin =\dimexpr\eqboxwidth{Des}+\labelsep}}%

\SetLabelAlign{parright}{\strut\smash{\parbox[t]\labelwidth{\raggedleft#1}}}

\usetikzlibrary{automata, positioning, decorations.pathmorphing, calc, quotes} 

\tikzset{
    every state/.append style={
        execute at begin node=$,
        execute at end node=$
    },
    initial text = 
}

\usepackage{thmtools}
\usepackage{thm-restate}
\declaretheorem{theorem}
\declaretheorem[sibling=theorem]{lemma,corollary}

\theoremstyle{definition}

\newtheorem{definition}[theorem]{Definition}

\crefname{fact}{fact}{facts}
\Crefname{fact}{Fact}{Facts}

\crefalias{constraint}{equation}
\crefname{constraint}{constraint}{constraints}
\Crefname{constraint}{Constraint}{Constraints}
\creflabelformat{constraint}{#2\textup{(#1)}#3} 
\crefalias{sentence}{equation}
\crefname{sentence}{sentence}{sentences}
\Crefname{sentence}{Sentence}{Sentences}
\creflabelformat{sentence}{#2\textup{(#1)}#3}
\crefalias{expression}{equation}
\crefname{expression}{expression}{expressions}
\Crefname{expression}{Expression}{Expressions}
\creflabelformat{expression}{#2\textup{(#1)}#3}

\NewDocumentEnvironment{delineate}{m}{\textcolor{cyan!70!black!}{> > > > Begin: #1 > > > >}}{\textcolor{red!70!black!}{< < < < End: #1 < < < <}}

\mathchardef\mhyphen="2D

\DeclarePairedDelimiter\paren\lparen\rparen

\newcommand{\oh}[1]{\ensuremath{\mathit{o}\paren*{#1}}}

\makeatletter
\newcommand*{\IfItalicsTF}{%
  \ifx\f@shape\my@test@it
    \expandafter\@firstoftwo
  \else
    \expandafter\@secondoftwo
  \fi
}
\newcommand*{\my@test@it}{it}
\makeatother

\newcommand{\contextsensitivemathrm}[1]{\IfItalicsTF{\mathit{#1}}{\mathrm{#1}}}

\newcommand\machineformat[1]{\ensuremath{\contextsensitivemathrm{#1}}}
\newcommand\langclassformat[1]{\ensuremath{\mathsf{#1}}}

\newcommand{\langformat}[1]{\ensuremath{\mathtt{#1}}}

\NewDocumentCommand{\DTIME}{ o m }{\langclassformat{DTIME\IfValueTF{#1}{\paren[#1]{#2}}{\paren*{#2}}}}

\NewDocumentCommand{\defineautomata}{ m m }{%
    \expandafter\NewDocumentCommand\csname#1fa\endcsname{}{\ensuremath{\machineformat{#1fa}}}%
    \expandafter\NewDocumentCommand\csname rt#1fa\endcsname{}{\ensuremath{\machineformat{rt\mhyphen#1fa}}}%
    \expandafter\NewDocumentCommand\csname o#1fa\endcsname{ o }{\ensuremath{\machineformat{1#1fa\IfValueT{##1}{\paren*{##1}}}}}%
    \expandafter\NewDocumentCommand\csname t#1fa\endcsname{ o }{\ensuremath{\machineformat{2#1fa\IfValueT{##1}{\paren*{##1}}}}}%
    \expandafter\NewDocumentCommand\csname#1tm\endcsname{}{\ensuremath{\machineformat{#1tm}}}%
    \expandafter\NewDocumentCommand\csname #2FA\endcsname{ d<> }{\ensuremath{\langclassformat{#2FA}\IfValueT{##1}{\paren*{##1}}}}%
    \expandafter\NewDocumentCommand\csname RT#2FA\endcsname{ d<> }{\ensuremath{\langclassformat{RT\mhyphen#2FA}\IfValueT{##1}{\paren*{##1}}}}%
    \expandafter\NewDocumentCommand\csname O#2FA\endcsname{ o d<> }{\ensuremath{\langclassformat{1#2FA}\IfValueTF{##1}{\paren*{##1\IfValueT{##2}{, ##2}}}{\IfValueT{##2}{\paren*{##2}}}}}%
    \expandafter\NewDocumentCommand\csname T#2FA\endcsname{ o d<> }{\ensuremath{\langclassformat{2#2FA}\IfValueTF{##1}{\paren*{##1\IfValueT{##2}{, ##2}}}{\IfValueT{##2}{\paren*{##2}}}}}%
}

\defineautomata{d}{D}
\defineautomata{n}{N}
\defineautomata{a}{A}
\defineautomata{p}{P}
\defineautomata{q}{Q}
\defineautomata{qc}{QC}

\newcommand{\ith}[2][th]{\ensuremath{#2}\text{#1}}

\newcommand\setneg[1]{\ensuremath{\overline{#1}}}

\newcommand{\PL}{\langformat{PAD}}

\newcommand{\orderedeq}{\ensuremath{EQ}}

\newcommand{\acc}{\ensuremath{\mathrm{acc}}}
\newcommand{\rej}{\ensuremath{\mathrm{rej}}}

\newcommand{\qacc}{\ensuremath{q_{\acc}}}
\newcommand{\qrej}{\ensuremath{q_{\rej}}}
\newcommand{\sacc}{\ensuremath{s_{\acc}}}
\newcommand{\srej}{\ensuremath{s_{\rej}}}

\newcommand{\lend}{\ensuremath{\rhd}}
\newcommand{\rend}{\ensuremath{\lhd}}

\newcommand{\reject}{\emph{reject}}

\setlist{itemsep=0pt}

\newlist{observations}{enumerate}{1}
\setlist[observations]{
    label=\arabic{*}.,
    ref=\arabic{*}
}
\crefname{observationsi}{observation}{observations}
\Crefname{observationsi}{Observation}{Observations}

\newlist{differences}{enumerate}{1}
\setlist[differences]{
    label=\arabic{*}.,
    ref=\arabic{*}
}
\crefname{differencesi}{difference}{differences}
\Crefname{differencesi}{Difference}{Differences}

\newlist{strategies}{enumerate}{1}
\setlist[strategies]{
    label=Strategy \arabic{*}:,
    ref=\arabic{*},
    labelwidth=\widthof{Strategy 1:},
    leftmargin=\parindent+\labelwidth+\labelsep
}
\crefname{strategiesi}{strategy}{strategies}
\Crefname{strategiesi}{Strategy}{Strategies}

\newlist{turingenum}{enumerate}{1}
\setlist[turingenum]{
    noitemsep,
    labelsep=.5em,
    leftmargin=1em+\parindent,
    labelwidth=1em,
    label=(\Roman{*}),
    ref=\Roman{*}
}

\crefname{turingenumi}{Stg.}{Stgs.}
\Crefname{turingenumi}{Stage}{Stages}

\SetLabelAlign{Center}{\hss#1\hss}

\newcommand{\narrowfont}[3]{\scalebox{#1}[1.0]{#3}}
\newcommand{\turinglabelformat}[1]{\narrowfont{0.85}{-20}{\scriptsize#1}}
\newlength{\turinglabelgap}

\ExplSyntaxOn
\DeclareExpandableDocumentCommand{\IfNoValueOrEmptyTF}{mmm}
 {
  \IfNoValueTF{#1}{#2}
   {
    \tl_if_empty:nTF {#1} {#2} {#3}
   }
 }
\ExplSyntaxOff

\NewDocumentEnvironment{turing}{ O{} m m }
 {\IfNoValueOrEmptyTF{#1}{\setlength{\turinglabelgap}{0em}}{\setlength{\turinglabelgap}{0.5em}}\small\begin{enumerate}[labelsep=0pt,align=left,parsep=0pt,leftmargin=\widthof{\turinglabelformat{#1}}+\turinglabelgap, 
 listparindent=0pt] 
  \item[]\ignorespaces#3\\[0.5em]
  \begin{turingenum}[
    nosep,
    align=Center,
    labelwidth=\widthof{\turinglabelformat{#1}},
    labelsep=\turinglabelgap,
    leftmargin=0em 
  ]}
 {\unskip\end{turingenum}\end{enumerate}}

\usepackage{crossreftools}
\makeatletter
\newcommand{\optionaldesc}[3]{%
  \phantomsection
#1\protected@edef\@currentlabel{#1}\protected@edef\cref@currentlabel{%
    [#3][\arabic{#3}][\cref@result]%
    #1%
  }\label{#2}%
}
\makeatother

\NewDocumentCommand{\defineturingitem}{ m m }{%
    \expandafter\NewDocumentCommand\csname#1item\endcsname{ o o m }{\IfValueTF{##1}{\item[\turinglabelformat{\optionaldesc{##1}{\IfValueTF{##2}{##2}{stg:##1}}{turingenumi}}]\begin{adjustwidth}{#2}{0pt}\ignorespaces##3\end{adjustwidth}}{\item[]\begin{adjustwidth}{#2}{0pt}\ignorespaces##3\end{adjustwidth}}}%
}

\defineturingitem{t}{0em}
\defineturingitem{tt}{1em}
\defineturingitem{ttt}{2em}
\defineturingitem{tttt}{3em}
\defineturingitem{ttttt}{4em}

\makeatletter
\def\squiggly{\bgroup \markoverwith{\lower3.9\p@\hbox{\sixly \scalebox{1.2}[0.65]{\char58}}}\ULon}
\makeatother

\def\mysout{\leavevmode\bgroup\def\ULthickness{1pt}\ULdepth=-.4ex\ULset}

\newcommand{\stkout}[1]{\begingroup\ifmmode\text{\mysout{\ensuremath{#1}}}\else\mysout{#1}\fi\endgroup}

\newcommand{\utkanworry}[1]{\textcolor{red!45!black!90}{\ifmmode\smash[b]{\squiggly{#1}}\else\squiggly{#1}\fi}}

\renewcommand{\textvisiblespace}[1][.7em]{%
  \makebox[#1]{%
    \kern.07em
    \vrule height.5ex
    \hrulefill
    \vrule height.5ex
    \kern.07em
  }
}

\DeclareMathAlphabet{\mathsl}{\encodingdefault}{\familydefault}{m}{sl}
\SetMathAlphabet{\mathsl}{bold}{\encodingdefault}{\familydefault}{bx}{sl}

\usepackage{fancyhdr}
\author{A. C. Cem Say\\\href{mailto:say@bogazici.edu.tr}{\texttt{say@bogazici.edu.tr}}}
\title{Time hierarchies for sublogarithmic-space quantum computation}
\date{\small \itshape
  Department of Computer Engineering,
  Bo\u{g}azi\c{c}i University,
  Bebek 34342,
  \.{I}stanbul,
  T\"{u}rkiye
  }

\newcommand{\ie}{i.e.}

\NewDocumentCommand{\padlang}{ e{_} m }{\ensuremath{\PL_{\IfValueT{#1}{#1}}\paren*{#2}}}

\NewDocumentCommand{\padpair}{ s m }{%
    \IfBooleanTF{#1}%
    {\ensuremath{\paren*{\padlang{#2}, \padlang{\setneg{#2}}}}}%
    {\ensuremath{\paren{\padlang{#2}, \padlang{\setneg{#2}}}}}%
}

\NewDocumentCommand{\branch}{ m o }{\makebox{\textsc{[#1]}\IfValueT{#2}{ Branch #2:}}}
\NewDocumentCommand{\branchpr}{ m o }{\branch{\,Probability: \ensuremath{#1}\,}[#2]}
\NewDocumentCommand{\branchperc}{ m o }{\branch{#1\% prob.}[#2]}

\NewDocumentCommand{\vartextvisiblespace}{ O{.7em} O{.7ex} }{%
  \makebox[#1]{%
    \kern.07em
    \vrule height#2
    \hrulefill
    \vrule height#2
    \kern.07em
  }
}

\newcommand{\estring}{\ensuremath{\lambda}}
\newcommand{\blanksymb}{\ensuremath{\text{\vartextvisiblespace}}}
\NewDocumentCommand{\Qpub}{ e{_} }{\ensuremath{Q_{\textnormal{pub}\IfValueT{#1}{,#1}}}}
\NewDocumentCommand{\Qpri}{ e{_} }{\ensuremath{Q_{\textnormal{pri}\IfValueT{#1}{,#1}}}}
\NewDocumentCommand{\Qcom}{ e{_} }{\ensuremath{Q_{\textnormal{com}\IfValueT{#1}{,#1}}}}
\NewDocumentCommand{\bpub}{ e{_} }{\ensuremath{b_{\textnormal{pub}\IfValueT{#1}{,#1}}}}
\NewDocumentCommand{\bpri}{ e{_} }{\ensuremath{b_{\textnormal{pri}\IfValueT{#1}{,#1}}}}

\makeatletter

\def\@testdef #1#2#3{%
  \def\reserved@a{#3}\expandafter \ifx \csname #1@#2\endcsname
 \reserved@a  \else
\typeout{^^Jlabel #2 changed:^^J%
\meaning\reserved@a^^J%
\expandafter\meaning\csname #1@#2\endcsname^^J}%
\@tempswatrue \fi}

\begin{document}

\maketitle

\begin{abstract}
\noindent We present new results on the landscape of problems that can be solved by quantum Turing machines (QTM's) employing severely limited amounts of memory. In this context, we demonstrate two infinite time hierarchies of complexity classes within the ``small space'' regime: For all $i\geq 0$, there is a language that can be recognized by a constant-space  machine in $2^{O(n^{1/2^i})}$ time, but not by any sublogarithmic-space  QTM in $2^{O(n^{1/2^{i+1}})}$ time. For quantum machines operating within $o(\log \log n)$ space, there exists another hierarchy, each level of which corresponds to an expected runtime of $2^{O((\log n)^i)}$ for a different positive integer $i$. We also improve a quantum advantage result, demonstrating a language that can be recognized by a polynomial-time constant-space QTM, but not by any classical machine using $o(\log \log n)$ space, regardless of the time budget. The implications of our findings for quantum space-time tradeoffs are discussed.
  \vspace{1em}

\end{abstract}

\section{Introduction}\label{sec:intro}
Quantum computation presents several interesting challenges, both theoretical and technological. Although fast quantum algorithms have been discovered for some important problems, this does not in itself constitute proof of the superiority of quantum computers  over classical ones in cases (like integer factorization \cite{S97})  where  the problem in question is not known to be classically difficult. One area in which several unconditional quantum advantage results have already been established is the study of computation under small (e.g. sublogarithmic) space bounds. Examining the capabilities and limitations of such memory-constrained machines is also relevant given the current state of quantum hardware development.

It has been proven \cite{W03} that, for any space-constructible function $S(n) \in \Omega(\log n)$, any language that is recognized with bounded error by a quantum Turing machine (QTM) using $S(n)$ space can also be recognized by
a deterministic Turing machine using only quadratically more space. The quantum advantage becomes apparent when one focuses on the case of constant space. For any string $w$, let $w^R$ denote the reverse of  $w$. We will be considering the following languages.
\begin{equation}\label{eq:eq}
    \textit{EQ}=\Set{\mathtt{a}^n \mathtt{b}^n | n\geq 0}
\end{equation}
\begin{equation}\label{eq:pal}
    \textit{PAL}=\Set{w | w\in\{\mathtt{a},\mathtt{b}\}^*, w=w^R}
\end{equation}

Ambainis and Watrous \cite{AW02} and Remscrim \cite{R20} have proven the following important facts regarding the ways in which two-way quantum finite automata  outperform their classical counterparts: 
\begin{itemize}
    \item There exists a constant-space quantum machine that recognizes   \textit{EQ} in polynomial expected time, whereas no probabilistic Turing machine that uses $o(\log \log n)$ space can recognize any nonregular language in  polynomial time. \cite{AW02,R20,DS90} Stated formally in terms of simultaneous time-space complexity classes, 
    \begin{equation}
        \mathsf{BQTISP}(n^{O(1)},O(1))\not\subseteq \mathsf{BPTISP}(n^{O(1)},o(\log \log n)).\label{eq:awpoly}
    \end{equation}
    \item There exists a constant-space quantum machine that recognizes  \textit{PAL} in  expected time $2^{O(n)}$, whereas no probabilistic Turing machine that uses $o(\log n)$ space can recognize \textit{PAL}, regardless of the allowed time budget. \cite{AW02,DS92} Formally, 
    \[
    \mathsf{BQTISP}(2^{O(n)},O(1))\not\subseteq \mathsf{BPSPACE}(o(\log n)).
    \]
\end{itemize}
The result concerning \textit{PAL} was later complemented by Remscrim, who proved \cite{R21} that no quantum Turing machine running in space $o(\log n)$ and expected time $2^{n^{1-\Omega(1)}}$ can recognize that language with bounded error, establishing  that the class of languages recognizable by small-space quantum machines in polynomial time is properly contained in the class associated with exponential time.

This paper presents  new results that refine the picture described above, by demonstrating the existence of two  time hierarchies of complexity classes associated with sublogarithmic-space quantum computation and proving a  statement stronger than Equation \ref{eq:awpoly} about the power of polynomial-time constant-space quantum machines. The rest of the paper is structured as follows: After the relevant technical introduction in Section \ref{sec:defs}, Section \ref{sec:hier1} considers the ``upper end'' of the spectrum of intermediate time bounds that lie strictly between polynomial and exponential time, and concludes  (Theorem \ref{thm:hier1}) that, 
    for all $i \geq 0$, 
    \[
    \mathsf{BQTISP}(2^{O(n^{1/2^i})},o(\log n))\supsetneq \mathsf{BQTISP}(2^{O(n^{1/2^{i+1}})},o(\log n)).
\]

In Section \ref{sec:hier2}, we focus on the lower end of the intermediate regime mentioned above, and present a quasi-polynomial time hierarchy for quantum computation in $o(\log \log n)$ space. Specifically, we prove Theorem \ref{thm:hier2}, which states
that,
    for all $i \geq 1$, 
\begin{equation}
    \mathsf{BQTISP}(2^{O((\log n)^i)},o(\log \log n))\subsetneq \mathsf{BQTISP}(2^{O((\log n)^{i+1})},o(\log \log n)).\label{eq:h2}
\end{equation}

A well-known theorem of Freivalds  \cite{F81} shows the existence of a \textit{classical} constant-space  algorithm that recognizes \textit{EQ} in $2^{O(n)}$ time. Whether every language recognizable in polynomial time by some constant-space quantum algorithm can also be recognized by a two-way probabilistic finite automaton in exponential time  has been an open question until now. We use a member of the family of languages defined to establish the hierarchy of Theorem \ref{thm:hier2} (Equation \ref{eq:h2}) to prove Theorem \ref{thm:adv}, which answers that question negatively, demonstrating that
\[
\mathsf{BQTISP}(n^{O(1)},O(1))\not\subseteq\mathsf{BPSPACE}(o(\log \log n)).
\] 
Section \ref{sec:conc} includes a discussion of time-space tradeoffs implied by our results, and an open question. 

\section{Definitions and background}\label{sec:defs}

We will use classical and quantum versions of the Turing machine model to facilitate the measurement and comparison of the memory requirements of the problems to be considered in this paper. Both classical and quantum Turing machines are defined to have a read-only input tape that they access with a two-way head.  The transition function of every such machine is restricted so that the input head never attempts to move off the tape area containing the string $\lend w \rend$, where $\lend$ and $\rend$ are special endmarker symbols, and  $w$ is the input string. Work tapes are used to measure memory consumption in both models, as will be detailed below.

\begin{definition}\label{def:ptm}
A \emph{probabilistic Turing machine (PTM)}  is a 7\=/tuple  $\paren*{Q,\Sigma,K,\delta,q_0,q_{acc},q_{rej}}$, where
\begin{itemize}
    \item $Q$ is the finite set of states,
    \item $\Sigma$ is the finite input alphabet, not containing  $\lend$ and $\rend$, 
    \item $K$ is the finite work alphabet, including the blank symbol \blanksymb,
    \item 
    $\delta: \paren*{Q \setminus \Set{q_{acc},q_{rej}}} \times \paren*{\Sigma \cup \Set{\lend, \rend}} \times K  \to \mathcal{P}(Q \times K \times \Set{-1,0,1}^2)$, such that, for each fixed $q$, $\sigma$, and $\kappa$, 
    $\lvert \delta(q,\sigma,\kappa)\rvert \in \Set{1,2}$, is the transition function,
    \item $q_0 \in Q$ is the start state, and
    \item $q_{acc}, q_{rej} \in Q$, such that $q_{acc}\neq q_{rej}$, are the accept and reject states, respectively.
\end{itemize}
\end{definition}

A PTM has a single read\-/write work tape, which is initially blank. The computation of a  PTM $M$ on input $w \in \Sigma^*$ begins with the input head positioned on the first symbol of   the string $\lend w \rend$, and the machine in state $q_0$. If $M$ is in state $q_{acc}$ ($q_{rej}$), it halts and accepts (rejects). If $M$ is in some state $q  \notin \Set{q_{acc},q_{rej}}$ with the input head scanning some symbol $\sigma$ at the \ith{i} tape position and the work head scanning some symbol $\kappa$ at the \ith{j} tape position,  a member 
$(q',\kappa',d_I,d_W)$ of $\delta(q,\sigma,\kappa)$ is selected with probability $1/\lvert \delta(q,\sigma,\kappa)\rvert$, and $M$  switches to state $q'$, sets the scanned work symbol to $\kappa'$, and locates the input and work heads at the \ith{(i+d_I)} and \ith{(j+d_W)} positions, respectively.
We assume that $\delta$ is defined such that the work head never moves to the left of its initial position.

A PTM is said to \emph{run within space} $S(n)$ if, on any input of length $n$, at most $S(n)$ work tape cells are scanned throughout the computation.  A computation performed by a machine that runs within $O(1)$ space can also be performed by a machine which does not use its work tape at all, and encodes all possible memory contents in its state set. This constant\-/space model is known as the \emph{two\-/way probabilistic finite automaton (2PFA)}.

A (classical or quantum) Turing machine $M$ is said to \textit{recognize} a language $L$ with error bound $\varepsilon$ if there exists a number $\varepsilon<\frac{1}{2}$ such that, for every input string $w\in \Sigma^*$,
\begin{itemize}
    \item if $w\in L$, $M$ accepts $w$ with probability at least $1-\varepsilon$, and
    \item if $w\notin L$, $M$ accepts $w$ with probability at most $\varepsilon$.
\end{itemize}

 $\mathsf{BPTISP}(T(n),S(n))$ denotes the class of languages recognizable with bounded error by PTM's running within space $S(n)$ and in expected time $T(n)$. When only the space complexity is considered, without regard to runtime, we use class names of the form $\mathsf{BPSPACE}(S(n))$.

A  deterministic Turing machine is   simply a PTM whose transition function is restricted so that  
    $\lvert \delta(q,\sigma,\kappa)\rvert =1$ for all $q$, $\sigma$ and $\kappa$.\footnote{Restricting a deterministic Turing machine so that it never uses its work tape yields the constant-space model known as the \textit{two-way deterministic finite automaton (2DFA)}.}
A quantum Turing machine  
\cite{W03} can be viewed as a deterministic Turing machine augmented by adding a finite quantum register, a second, quantum work tape, each cell of which holds a qubit, and a classical  measurement register of fixed size, which will be used to store the  outcome of the latest measurement of the quantum part. As we will see, 
the randomness in the behavior of a QTM is a result of the operations performed on the quantum part of the machine at intermediate steps.

\begin{definition}\label{def:qtm}
A \emph{quantum Turing machine (QTM)}  is an 11-tuple  $\paren*{Q,\Sigma,K,T,k,\delta_Q,\delta_C,q_0,\tau_0,q_{acc},q_{rej}}$, where
\begin{itemize}
    \item $Q$ is the finite set of classical states,
    \item $\Sigma$ is the finite input alphabet, not containing  $\lend$ and $\rend$, 
    \item $K$ is the finite  work alphabet, including the blank symbol \blanksymb,
    \item $T$ is the finite set of possible values of the measurement register,
    \item $k$, a positive integer, is the number of qubits in the quantum register,
    \item $\delta_Q: \paren*{Q \setminus \Set{q_{acc},q_{rej}}} \times \paren*{\Sigma \cup \Set{\lend, \rend}} \times K \to \Delta$, where $\Delta$ is the set of all selective quantum operations \cite{W03} that have operator-sum representations consisting of matrices of algebraic numbers, is the quantum transition function, 
    \item $\delta_C: \paren*{Q \setminus \Set{q_{acc},q_{rej}}} \times \paren*{\Sigma \cup \Set{\lend, \rend}} \times K \times T  \to Q \times K \times \Set{-1,0,1}^3$ is the classical transition function, which determines the updates to be performed on the classical components of the machine, 
    \item $q_0 \in Q$ is the initial   state,
    \item $\tau_0 \in T$ is the initial value of the measurement register, and
    \item $q_{acc}, q_{rej}\in Q$, such that $q_{acc} \neq q_{rej}$, are the accept and reject states, respectively.
\end{itemize}
\end{definition}

At the start of  computation,  a QTM $M$ is in state $q_0$, the measurement register contains $\tau_0$, the input head  is located on the left endmarker $\lend$,  the classical work tape contains all blanks, and all qubits of the finite quantum register and the quantum work tape are in their zero states.  Every computational step of $M$ consists of a \textit{quantum transition}, which operates on the quantum part of the machine and updates the measurement register, and  a \textit{classical transition}, which updates the machine state, classical work tape, and the three head positions, based on the presently scanned input and classical work tape symbols, and the value in the measurement register: If $M$ is in some non-halting state $q$, with the input head scanning the symbol $\sigma$ and the classical work tape head scanning the symbol $\kappa$, the selective quantum operation 
$\delta_Q(q,\sigma,\kappa)$
acts on the quantum register and the qubit presently scanned by the quantum work tape head.\footnote{In Watrous' original definition \cite{W03} of this QTM model, the quantum transitions depend only on the current state $q$ of the machine. The version we use here is equivalent in power (each step of our machines can be simulated in two steps by Watrous' machines), and is preferred because its specialization to the formal definition of the 2QCFA model \cite{AW02} (to be introduced shortly) is trivial.} This action may be a unitary transformation or a measurement,\footnote{Other nonunitary quantum operations are also allowed by the formalism, but the concrete quantum algorithms to be presented in this paper only need unitary transformations and projective measurements. A detailed exposition of selective quantum operations is presented in \cite{W03}.} resulting in those $k+1$ qubits evolving, and some output $\tau$ being stored in the measurement register with an associated probability. The ensuing classical transition  $\delta_C(q,\sigma,\kappa, \tau)=(q',\kappa',d_I,d_W,d_Q)$ updates the machine state to $q'$, changes the presently scanned classical work tape symbol to $\kappa'$, and moves the input, classical work and quantum work heads in the directions indicated by the increments $d_I$, $d_W$, and $d_Q$, respectively.
Computation halts with acceptance (rejection) when $M$ enters the  state $q_{acc}$ ($q_{rej}$).

As with PTM's, we focus on QTM's that never move either work tape head to the left of its initial position.
A QTM is said to \textit{run within space} $S(n)$ if, on any input of length $n$, at most $S(n)$ cells are scanned on both the classical and quantum work tapes during the computation.

Analogously to the classical version, $\mathsf{BQTISP}(T(n),S(n))$ denotes the class of languages recognizable with bounded error by QTM's running within space $S(n)$ and in expected time $T(n)$.

\begin{definition}\label{def:dissimilar}
\cite{DS90} For any language $L$ and number $n>0$, two strings $w$ and $w'$ are said to be  \emph{$n$\-/dissimilar}, written $w \not\sim_{L,n} w'$, if $\abs{w} \leq n $, $\abs{w'} \leq n$, and there exists a \emph{distinguishing string} $v$ with $\abs{wv} \leq n$, $\abs{w'v} \leq n$, and $wv \in L \iff w'v \notin L$. Let $D_{L}(n)$ be the maximum $d$ such that there exist $d$ distinct strings that are pairwise $\not\sim_{L,n}$.    
\end{definition}
 The following result, which  links the ``difficulty'' measure $D_{L}(n)$ defined above to the simultaneous time-space complexity of a QTM that recognizes the associated language $L$, is due to Remscrim. \cite{R21} This theorem will provide  an important tool in our proofs in the following sections.

\begin{theorem}\label{thm:remscrimQTMtradeoff}
Suppose \( L \in \mathsf{BQTISP}(T(n),S(n)) \), and suppose further that \( S(n) = o(\log \log D_L(n)) \). Then there exists a real number $ b_0 >0$ such that
\[
T(n) = \Omega\big(2^{-b_0 S(n)} D_L(n)^{2^{-b_0 S(n)}}\big).
\]
\end{theorem}

The QTM's that will be constructed to prove our main results are constant\-/space machines. Restricting the QTM model of Definition \ref{def:qtm} so that it never uses the work tapes yields the well\-/studied \textit{two-way finite automaton with  quantum and classical states (2QCFA)}, due to Ambainis and Watrous~\cite{AW02}. We will present our algorithms in Sections \ref{sec:hier1} and \ref{sec:hier2} as 2QCFA's, with no mention of either work tape.\footnote{The original 2QCFA definition in \cite{AW02} does not restrict the machines to use only algebraic numbers in the matrices specifying the quantum transitions. Our algorithms will conform to Definition \ref{def:qtm}, which requires all transition amplitudes to be algebraic numbers.} 







Our constructions in Sections \ref{sec:hier1} and \ref{sec:hier2} involve various subroutines that are obtained by adapting two basic 2QCFA algorithms. The first of these recognizes the language \orderedeq{} in polynomial expected time, and was discovered by Remscrim \cite{R20}.\footnote{Ambainis and Watrous were the first to present a polynomial-time 2QCFA for \orderedeq{} in \cite{AW02}, but their algorithm uses some amplitudes  that are not restricted to being algebraic, and thus does not fit our QTM definition.}
This machine, which we will call $M_{EQ}$, operates with \textit{negative one-sided bounded error}, i.e. it accepts members of \orderedeq{} with probability 1, and can be ``tuned'' so that it  rejects nonmembers with probability at least $1-\varepsilon$, for any desired small positive value of $\varepsilon$. We will adapt $M_{EQ}$ as the basis of two  subroutines  named \textsc{SameLength} and \textsc{TwiceAsLong}, which are employed in our algorithms in Sections \ref{sec:hier1} and \ref{sec:hier2} to compare the lengths of selected  subsequences of the input with each other. 

The second basic 2QCFA that we shall employ was  discovered by Ambainis and Watrous. \cite{AW02} This algorithm, named $M_{PAL}$, recognizes the language
\textit{PAL} with negative one-sided bounded error, where the error bound $\varepsilon$, i.e. the maximum allowed probability with which a nonmember of $PAL$ can be accepted, can be tuned down to any desired positive value, as described above for $M_{EQ}$. Unlike $M_{EQ}$, the expected runtime of $M_{PAL}$ on inputs of length $n$ is $2^{\Theta(n)}$, where the constant ``hidden'' in the big-$\Theta$ notation depends on $\varepsilon$. Remscrim \cite{R21} has proven that this runtime  is asymptotically optimal for constant-space machines, and we will be careful in Sections \ref{sec:hier1} and \ref{sec:hier2}  to call the subroutine \textsc{Pal} that performs this palindromeness check only on very short prefixes of the input to ensure that the expected runtimes of our algorithms remain subexponential.

The reader should note that our  results depend only on the \textit{existence} of  2QCFA's that recognize \textit{EQ} and \textit{PAL} with negative one-sided bounded error within the runtimes specified above, and not on the details of the ``inner workings'' of the associated machines. For instance, an alternative 2QCFA presented by Yakary{\i}lmaz and Say \cite{YS1001} also satisfies those conditions for \textit{PAL}, and can be used   in the role of $M_{PAL}$ in our algorithms without affecting the results.




\section{A time hierarchy for computation with $o(\log n)$ space}\label{sec:hier1}


We start by defining an infinite family of languages on the alphabet $\Sigma=\{\mathtt{a},\mathtt{b},1,\$\}$. In the following, $\Sigma_{\mathtt{ab}}=\{\mathtt{a},\mathtt{b}\}$. 

 Let $RL_0=\Sigma_{\mathtt{ab}}\Sigma_{\mathtt{ab}}^+$, and 
 define, for each $i\geq1$,
\begin{equation}
RL_i = \{w\$1(\mathtt{a}^{|w|}1)^{|w|-1} \mid w\in RL_{i-1}\}. 
\label[expression]{eqn:rldef} 
\end{equation}

Let us  examine the pattern 
imposed by the recursive definition above on some member, say, $s$, of $RL_i$.
$s$ consists of $i+1$ ``blocks'', separated from each other with the symbol $\$$.  
For any $j\in \{0,\dots , i\}$, let $s_j$ denote the unique prefix of $s$ that is a member of $RL_j$, i.e., $s_i=s$, and $s_0$, the leftmost block, is a string of length at least 2
on the alphabet $\Sigma_{\mathtt{ab}}$. 
For each  $j\geq 1$,  $s_j$'s length is related to the length of $s_{j-1}$ as
\begin{equation}
    |s_j|=|s_{j-1}|^2+|s_{j-1}|+1,
\label[expression]{eqn:srel}    
\end{equation}
because 
$s_j$  contains precisely $|s_{j-1}|$ occurrences of the 
symbol 1
that delimit $|s_{j-1}|-1$ ``segments'' of the form $\mathtt{a}^{|s_{j-1}|}$ in its rightmost block (that follows its rightmost $\$$).

\begin{lemma}\label{rootn}
  For any $i>0$, there exists a constant $a_i>0$ such that the following relations hold between the length $n$ of any string  $s\in RL_i$ and the length of the leftmost block $s_0$ of $s$:
  \begin{enumerate}
      \item $|s_0|<n^{1/2^i}$.
      \item  $|s_0|>a_i n^{1/2^i}$.
  \end{enumerate}
\end{lemma}
\begin{proof}
We use Equation \ref{eqn:srel}.
Since $|s_j|>|s_{j-1}|^2$ for all $j$, we have  $n=|s_i|>|s_0|^{2^i}$, yielding the upper bound $|s_0|<n^{1/2^i}$.

To obtain a lower bound, note that 
 there exists a constant $c>0$  such that  $|s_j|<c|s_{j-1}|^2$ for all $j$.
 Since it is easy to prove that  $|s_j|<c^{2^{j}-1}|s_0|^{2^j}$ for all $j$,
 we conclude that $|s_0|>a_i n^{1/2^i}$, where $a_i=\frac{1}{c^{1-\frac{1}{2^{i}}}}$.
\end{proof}





We now define a new language family. For each $i\geq 1$, let
\begin{equation}RPAL_i = \{s \mid s\in w\$\Sigma^* \cap RL_{i}, w\in PAL\}. 
\label[expression]{eqn:rpaldef} 
\end{equation} 

 The only difference between $RL_i$ and $RPAL_i$ is that Equation \ref{eqn:rpaldef} imposes the additional condition that the leftmost block of each member of the language should be a palindrome. Note that $RPAL_i$ can be seen as a ``padded'' version of the language $PAL \cap \Sigma_{\mathtt{ab}}\Sigma_{\mathtt{ab}}^+$, where the padding has a specific internal syntax. 
As we will demonstrate shortly,  a 2QCFA can check whether its input is in $RL_i$ in polynomial expected time  with a very low probability of error, and  the time complexity of recognizing $RPAL_i$ is dominated by the remaining necessity to check whether the leftmost block of the input is a palindrome or not.



\begin{theorem}\label{thm:mrpal}
        For all $i\geq1$, $RPAL_i \in \mathsf{BQTISP}(2^{O(n^{1/2^{i}})},O(1))$.
\end{theorem}
\begin{proof}
Figure \ref{fig:qalg1} is a template which can be used to construct a 2QCFA $M_{RPAL_i}$ that recognizes $RPAL_i$ for any desired $i>0$. We will show that each $RPAL_i$ has expected runtime $2^{O(n^{1/2^{i}})}$.
In brief, the algorithm's main loop (M) spends polynomial expected time to check whether its input is a member of $RL_i$. This check is performed by running a subroutine on carefully designated subregions on the input tape to compare the lengths of several substring sequences. We will show that,  for input strings that are not in $RL_i$, the probability of the machine reaching stage (P) without rejection is exponentially small. That final stage runs the palindromeness checking algorithm on the leftmost input block, whose length can exceed the bound established in Lemma \ref{rootn} with only that minuscule probability, which will guarantee that the overall expected runtime is subexponential. The parameters $\varepsilon$ (a rational number) and $k_{\varepsilon}$ (a positive integer) are used for tuning the error bound and runtime of the algorithm. We now describe the details.

\begin{figure}[htb!]
    \caption{$M_{RPAL_i}$}
    \label{fig:qalg1} 
    \begin{turing}[(RW)]{V}{\textit{  }}
        \titem[(R)]{If the input is not of the  form 
$\Sigma_{\mathtt{ab}}\Sigma_{\mathtt{ab}}^+(\$1(\mathtt{a}^+1)^+)^+$, 
        \reject.}
        \titem[(M)]{
            Repeat ad infinitum:}
        \ttitem{For each integer $j$ from $i$ down to and including 1, do the following:}
        \tttitem{\{Let $p_l$ denote the  position of the left endmarker.\} }
        \tttitem{\{Let $p_m$ denote the position of the $j$th $\$$ symbol from the left.
        \}
        }
        \tttitem{\{Let $p_r$ denote the position of the $(j+1)$st $\$$ symbol (or the right endmarker, if $j=i$).\}}
        \tttitem[(C1)]{Run \textsc{SameLength}$(\langle \Sigma,(p_l,p_m)\rangle,\langle\{1\},(p_m,p_r)\rangle,\varepsilon)$.}
        \tttitem{\{Let $p_r$  denote the position of the 
        second occurrence of the symbol 1 
        after the $j$th $\$$ symbol in the input.\}}
         \tttitem{While $p_r$ is the position of
        a $1$ symbol, do the following:
            }
        \ttttitem[(C2)]{    
            Run \textsc{SameLength}$(\langle\Sigma,(p_l,p_m)\rangle,\langle\{\mathtt{a}\},(p_m,p_r)\rangle,\varepsilon)$.}
        \ttttitem{    
            \{Update $p_r$ to denote the position of the 
            nearest non-$\mathtt{a}$ symbol on the tape to the right of its present value.\}
            }        
            \ttttitem{
            \{Update $p_m$ to denote the position of the 
            nearest 
            occurrence of the symbol 1 
            located to the left of $p_r$.\}} 
        \ttttitem{\{Update $p_l$ to denote the position of the 
            nearest 
            occurrence of the symbol 1 
            located to the left of $p_m$.\}}
        \ttitem[(RW)]{Move the tape head to the leftmost input symbol.}
        \ttitem{While the currently scanned symbol is not an endmarker, do the following:}
        \tttitem{Simulate a classical coin flip. If the result is ``heads'', move  one cell to the right. Otherwise, move  one cell left.}
        \ttitem{If the random walk ended at the right endmarker, do the following:}
        \tttitem{Simulate $k_{\varepsilon}$ coin flips. If all results are ``heads'', exit loop (M).}
        \titem[(P)]{
            Move the input head to the left endmarker and run \textsc{Pal}$($\$$,\varepsilon)$.}
    \end{turing}
\end{figure}

The regular expression check (R) at the start of the algorithm can be performed deterministically in linear time. Each iteration of the main loop (M) goes through the integers from $i$ down to 1 to check whether the input satisfies the following conditions imposed by the definition in Equations \ref{eqn:rpaldef} and \ref{eqn:rldef}:  
\begin{enumerate}
    \item For each $j$, every member of $RPAL_j$ should contain precisely $|w|$ occurrences of the 
    symbol 1
    to the right of the 
    $j$th 
    $\$$ symbol 
    from the left,
    where $w$ is the prefix of the input up to that $\$$ symbol. This is checked by stage (C1), as will be explained below.
    \item For each $j$, every member of $RPAL_j$ should contain precisely $|w|$ occurrences of the symbol $\mathtt{a}$ sandwiched between every consecutive pair of  
    1's following its  $j$th $\$$ symbol, where  $w$ is defined as above.  
This is checked by stage (C2).
\end{enumerate}

As mentioned in Section \ref{sec:defs}, we use adapted versions of the polynomial-time 2QCFA algorithm $M_{EQ}$ that recognizes the language $EQ=\Set{\mathtt{a}^n \mathtt{b}^n | n\geq 0}$ to perform the checks described above. Note that, whenever any  2QCFA that recognizes $EQ$ runs on an input of the form $\mathtt{a}^+\mathtt{b}^+$, it can be viewed as reporting whether the lengths of two disjoint subsequences of its input (which  have been placed on the tape in such a manner that it can distinguish where the subsequence on the left ends and the one on the right starts) are equal or not. In the Appendix, we describe how $M_{EQ}$ can be modified to obtain submachines that perform the more general subsequence length comparison tasks indicated by the ``subroutine calls'' of the form \textsc{SameLength}$(\langle S_{\mathbf{L}},I_{\mathbf{L}}\rangle,\langle S_{\mathbf{R}},I_{\mathbf{R}}\rangle,\varepsilon)$ in Figure \ref{fig:qalg1}.

A \textsc{SameLength} submachine conforming to the parameter list $(\langle S_{\mathbf{L}},I_{\mathbf{L}}\rangle,\langle S_{\mathbf{R}},I_{\mathbf{R}}\rangle,\varepsilon)$  runs an appropriately adapted version of $M_{EQ}$ to compare the number of occurrences of members of set $S_{\mathbf{L}}$ within the tape region corresponding to the position interval $I_{\mathbf{L}}$ with the  number of occurrences of members of  $S_{\mathbf{R}}$ within the tape region corresponding to interval $I_{\mathbf{R}}$, so that the error (i.e., the probability that $M_{EQ}$ misclassifies its input when the two subsequences are of unequal length) is bounded by the  number $\varepsilon>0$. The sets $S_{\mathbf{L}}$ and $S_{\mathbf{R}}$ can contain individual members of the input alphabet, or longer substrings. The position intervals $I_{\mathbf{L}}$ and $I_{\mathbf{R}}$ are delimited by 
three ``signposts'', denoted  $p_l$, $p_m$, and $p_r$ in the pseudocode in Figure \ref{fig:qalg1}.\footnote{For example, in this parlance, the original 2QCFA $M_{EQ}$ running on a string of the form $\mathtt{a}^+\mathtt{b}^+$ can be viewed as operating within the two intervals $(p_l,p_m)$ and $[p_m,p_r)$, where $p_l$, $p_m$, and $p_r$ are the positions of the left endmarker, the leftmost $\mathtt{b}$, and the right endmarker, respectively.} As detailed in the Appendix, the submachine realizing the \textsc{SameLength} procedure simulates $M_{EQ}$ on a virtual tape containing the string $\lend\mathtt{a}^j \mathtt{b}^k \rend$, where $j$ (resp. $k$) is the length of the left (resp. right) subsequence indicated by its parameters.

If $M_{EQ}$ would enter its reject state as a result of the length comparison dictated by these parameters, \textsc{SameLength}, and thereby  also $M_{RPAL_i}$, rejects and halts. If, on the other hand, the execution of $M_{EQ}$ is seen to be led to its accept state, \textsc{SameLength} simply returns control to the next instruction in $M_{RPAL_i}$. 
The probability that $M_{RPAL_i}$ rejects its input string $s$ in a single iteration of the main loop (M) is 0 if $s\in RPAL_i$, and at least $1-\varepsilon$ if $s\notin RL_i$.\footnote{The probability that $M_{RPAL_i}$ fails to reject a string $s\notin RL_i$ is maximized if $s$ has a single ``defect'' that can be caught by only one of the many \textsc{SameLength} calls executed in that iteration.}

To reduce the error bound to an exponentially small amount, $M_{RPAL_i}$ repeats the above-mentioned control sequence polynomially many times by executing a ``biased coin throw'' implemented by a random walk  within each iteration of the main loop, and exiting only when the less likely outcome of that coin is observed.  It is well known \cite{DS92,AW02} that a random walk as described in stage (RW) of 
 Figure \ref{fig:qalg1}  ends at the right endmarker with probability $1/(n+1)$ on inputs of length $n$. Considering the probability of obtaining ``heads'' in all the $k_{\varepsilon}$ additional fair coin tosses as well, the expected number of iterations of loop (M) is $2^{k_{\varepsilon}}(n+1)$ (unless a \textsc{SameLength} procedure call causes the machine to reject earlier). We conclude that, when running on an input string $s\notin RL_i$, loop (M) of $M_{RPAL_i}$ will fail to reject, and the machine will  therefore proceed to stage (P), with a  probability  $p_{bad}$ whose expected value is at most $\varepsilon^{2^{k_{\varepsilon}}(n+1)}$, where $n=|s|$. 
 
 Each iteration of (M) consists of a constant number of sweeps of the input, with each sweep comprising  at most $n$ calls of \textsc{SameLength}, each of which run in polynomial time. The random walk  has an expected runtime of $O(n)$. \cite{GS24} We conclude that the total expected runtime of (M) is polynomially bounded.

In its final stage (P), $M_{RPAL_i}$ calls a procedure based on Ambainis and Watrous' $M_{PAL}$ algorithm, which was also introduced in Section \ref{sec:defs}: A procedure call of the form \textsc{Pal}$(D,\varepsilon)$ operates only on the prefix of the input string that precedes the first occurrence of the substring $D$ in the input, feeding only the members of $\Sigma_{\mathtt{ab}}$ that occur in that prefix to $M_{PAL}$, skipping over any other symbols.\footnote{Although $M_{RPAL_i}$ would not encounter any nonmember of  $\Sigma_{\mathtt{ab}}$ in the first block of its input in stage (P), this ability to skip over other symbols will be useful when we employ the \textsc{Pal} procedure in another algorithm in Section \ref{sec:hier2}.} The basic algorithm $M_{PAL}$ is set up so that its error bound is $\varepsilon$. $M_{RPAL_i}$ terminates by acceptance or rejection as dictated by the \textsc{Pal} procedure at the end of stage (P).

Before analyzing the overall runtime of $M_{RPAL_i}$, which will be seen to be dominated by the final stage, we describe how to set the parameters  $\varepsilon$ and $k_{\varepsilon}$. $\varepsilon$ can simply be set to any desired positive rational number that will be the final error bound of  $M_{RPAL_i}$ in its recognition of $RPAL_i$. As mentioned in Section \ref{sec:defs} and shown in \cite{AW02}, there exists a number $c_{\varepsilon}>0$ (which depends on this choice of $\varepsilon$) so that the runtime of $M_{PAL}$ on inputs of sufficiently large  length $n$ is at most $2^{c_{\varepsilon}n}$.  We set $k_{\varepsilon}$ to the smallest positive integer that satisfies the relation $\varepsilon^{2^{k_{\varepsilon}}(n+1)}2^{c_{\varepsilon}n}\in o(1)$.

By Lemma \ref{rootn}, if the input string is in $RL_i$, the length of its first block, which is the ``input'' for the \textsc{Pal} procedure in stage (P), is less than $n^{1/2^i}$. In this case, stage (P) will have an expected runtime of at most $2^{c_{\varepsilon}n^{1/2^i}}$. If the input  is not in $RL_i$, it may be the case that the first block is much longer, and there is a nonzero probability that the execution of \textsc{Pal} on that block will exceed the runtime bound stated above. Fortunately, this probability nears 0 as $n$ grows. Calculating the expected runtime of $M_{RPAL_i}$ in that case, we obtain the upper bound 
\[n^{O(1)}+(1-p_{bad})\times 0+p_{bad}\times 2^{c_{\varepsilon}n}\leq n^{O(1)}+ 
\varepsilon^{2^{k_{\varepsilon}}(n+1)}2^{c_{\varepsilon}n},
\]
which yields an overall runtime of $n^{O(1)}$, dominated by stage (M). We conclude that, for all $i\geq 1$, each  $M_{RPAL_i}$ runs in $2^{O(n^{1/2^i})}$ time and recognizes the  language $RPAL_i$ with negative one-sided bounded error $\varepsilon$.
\end{proof}





We are now ready to state our first hierarchy theorem.

\begin{theorem}\label{thm:hier1}
    For each $i \geq 0$, $\mathsf{BQTISP}(2^{O(n^{1/2^i})},o(\log n))\supsetneq \mathsf{BQTISP}(2^{O(n^{1/2^{i+1}})},o(\log n))$.
\end{theorem}

\begin{proof}
   The \textit{PAL} language witnesses the separation  for $i=0$, since it can be recognized \cite{AW02} by a 2QCFA  in $2^{O(n)}$ time,  and no QTM employing $o(\log n)$ space can recognize it \cite{R21} in $2^{n^{1-\Omega(1)}}$ time.
    
    For each $i>0$, it has already been established (Theorem \ref{thm:mrpal}) that $RPAL_i \in \mathsf{BQTISP}(2^{O(n^{1/2^{i}})},O(1))$. 
    It remains to show that 
    \[RPAL_i \notin \mathsf{BQTISP}(2^{O(n^{1/2^{i+1}})},o(\log n)).\]

  We consider the measure $D_{RPAL_i}(n)$, based on Definition \ref{def:dissimilar}. 
  For any positive even number $m$, let $n$ be the length of the members of $RPAL_i$ whose leftmost blocks are of  length $m$. By Lemma \ref{rootn}, there exists a constant $a_i>0$ such that $m>a_i n^{1/2^i}$. Let $W_n$ be the set of all strings of length $m/2$ on $\Sigma_{\mathtt{ab}}$.  
  
    Since every distinct pair of strings $w$, $w'$ in $W_n$ are distinguished by the string $w^{R}\$t$, where $t$ is the unique postfix that satisfies $ww^{R}\$t \in RPAL_i$, we conclude that 
    $D_{RPAL_i}(n)\geq |W_n|\geq 2^{\frac{a_i}{2} n^{1/2^i}}$ for infinitely many values of $n$.



Assume that $RPAL_i \in \mathsf{BQTISP}(2^{O(n^{1/2^{i+1}})},o(\log n))$. This would entail the existence of a  QTM, say, $M$, that recognizes $RPAL_i$ with bounded error in sublogarithmic space $S(n)$ and in expected runtime at most  $2^{e_0 n^{1/2^{i+1}}}$ for some positive constant $e_0$. 
Since $\log \log D_{RPAL_i}(n) \in \Omega(\log n)$, we have $S(n) \in o(\log \log D_{RPAL_i}(n))$. 
By Theorem \ref{thm:remscrimQTMtradeoff}, 
   it must then be the case that 
   \begin{equation}
    2^{e_0 n^{1/2^{i+1}}}\in \Omega\big(2^{-b_0 S(n)} 2^{\frac{a_i}{2} n^{1/2^i} 2^{-b_0 S(n)}}\big)    \label{eq:lbsub}
   \end{equation}
    for some positive constant $b_0$  and for infinitely many values of $n$.
   
The fact that $S(n) \in o(\log n)$ implies that, for any positive constant $b_1$, $2^{-b_0 S(n)}\geq n^{-b_1}$ for all sufficiently large $n$.
We choose $b_1$ so that $\frac{1}{2^i}-b_1>\frac{1}{2^{i+1}}$, and manipulate Equation   \ref{eq:lbsub} to obtain 
   \begin{equation*}
    2^{e_0 n^{1/2^{i+1}}}\in \Omega\big(2^{\frac{a_i}{4} n^{(1/2^i)-b_1}} \big).   
   \end{equation*}
   But $n^{1/2^{i+1}}\in o(n^{(1/2^i)-b_1})$, which leads to a contradiction and completes the proof. 
\end{proof}
\section{A 
time hierarchy  for computation with $o(\log \log n)$ space}\label{sec:hier2}

The time hierarchy we examined in the previous section was near the ``top end'' of the gap between the classes of problems that can be solved by 2QCFA's with polynomial time budgets, and those that require  exponential time budgets. The  hierarchy to be demonstrated in this section is placed near the ``bottom end'' of that gap. To achieve this result, we will again use a language family based on padded palindromes. There will be  two  main differences with the family studied in Section \ref{sec:hier1}: Firstly, the length of the padding in each member of the languages under consideration here will be much larger relative to the length of the palindromic prefix. Secondly, those palindromes in the prefixes will have to be punctuated by inserting additional separator substrings, as we describe below.

\paragraph{The punctuation procedure.} We will begin by defining an auxiliary family of languages that contain punctuated palindromes, whose lengths are perfect powers. For any integer $i>0$, let $c_i=\lceil \log(i+1)\rceil$, and let $m_i$ be the smallest  integer greater than 1 such that $m^{i-1}<2^m$ for all $m\geq m_i$.\footnote{Logarithms are to the base 2.} 
The language $PPAL_i$ is  the set of all strings obtained by applying a ``punctuation'' procedure, to be described below, on  all palindromes of length $m^i$  on the alphabet $\Sigma_{\mathtt{ab}}$, for all $m\geq m_i$.

$PPAL_1$ is defined to be the concatenation $(\textit{PAL} \cap \Sigma_{\mathtt{ab}}\Sigma_{\mathtt{ab}}^+)\Set{1}$.

For $i>1$,
view any palindrome of length $m^i$ as the concatenation of $m^{i-1}$ ``segments'' of length $m$. 
For instance,  a palindrome $p$,
whose length is $5^3=125$, has 25 such segments, i.e.,
\[p=s_1 s_2 s_3 s_4 s_5 s_6 s_7 s_8 s_9 s_{10} 
s_{11} s_{12} s_{13} 
s_{14} s_{15} 
s_{16} s_{17} 
s_{18} s_{19} s_{20} 
s_{21} s_{22} 
s_{23} 
s_{24} 
s_{25},
\]
where each $s_j\in \Sigma_{\mathtt{ab}}^5$. For every $i\geq1$ and $j \in \{1,\dots,i\}$, let cbin$_i(j)$ denote the string of length $c_i$ that is the binary representation of the integer $j$.\footnote{Note that these strings can contain leading zeros.}
Given any  $m^i$-symbol palindrome as exemplified above, the punctuation procedure creates a longer string by  inserting  a binary delimiter substring of the form cbin$_i(j)$  after each $m$-symbol segment, including the last one. 
For instance, the palindrome $p$ in our  example above would give rise to the string 
\tiny
\[w= s_1 01s_2 01s_3 01s_4 01s_5 \mathbf{10}s_6 01s_7 01s_8 01s_9 01s_{10} \mathbf{10}
s_{11} 01s_{12} 01s_{13} 01
s_{14}01 s_{15} \mathbf{10}
s_{16}01 s_{17} 01
s_{18}01 s_{19}01 s_{20} \mathbf{10}
s_{21} 01 s_{22} 01
s_{23} 01
s_{24} 01
s_{25}11
\]
\normalsize
in $PPAL_3$. We will now specify how the precise ordering of the numerical values of the binary delimiters to be inserted in a palindrome of a given length is determined in the general case.


\begin{definition}\label{def:pat}
    For all $i>1$ and  $m\geq m_i$, a sequence $S$ of binary delimiters is said to be  \textit{well-ordered for (i,m)} 
    if it satisfies the following  conditions:
    \begin{enumerate}
    \item All members of $S$ are elements of the  set $\Set{\text{cbin}_i(1),\dots,\text{cbin}_i(i)}$,
        \item The delimiter cbin$_i(i)$ appears exactly once in $S$, as its last element, and
        \item For all  $j\in \{2,\dots, i\}$, $S$ contains exactly $m-1$ occurrences of the delimiter  cbin$_i(j-1)$ in its prefix which ends with the leftmost occurrence of a delimiter of the form cbin$_i(l_0)$, and also between any two successive occurrences of any two delimiters  cbin$_i(l_1)$ and cbin$_i(l_2)$, where $l_0,l_1, l_2 \geq j$.
    \end{enumerate}
\end{definition}
 The sequence of binary delimiters in our example string $w$ is well-ordered for (3,5). All delimiters in the sequence are from the set $\Set{\text{cbin}_3(1),\text{cbin}_3(2),\text{cbin}_3(3)}=\Set{01,10,11}$, the sequence ends with the unique occurrence of the delimiter 11, and the third condition in Definition \ref{def:pat} is satisfied: For instance, we see that the sequence contains exactly $4$ occurrences of 10 (shown in bold in the description of $w$ above) in its prefix ending with 11, and it also contains $4$ occurrences of 01 between, say, the rightmost occurrence of  10 and the occurrence of 11.

\begin{lemma}
 For all $i>1$ and  $m\geq m_i$, there exists a unique sequence $S_{i,m}$ of binary delimiters that is  well-ordered for $(i,m)$. Furthermore, the length of $S_{i,m}$ is $m^{i-1}$.
\end{lemma}
\begin{proof}
We start by showing that any 
   sequence $S$ of binary delimiters which is  well-ordered for $(i,m)$ must have length $m^{i-1}$. 
   To see this, let sumdel$_S(j)$ denote the total number of occurrences of binary delimiters of the form cbin$_i(t)$,  where $t\geq j$, in $S$. 
We claim that sumdel$_S(j)=m^{i-j}$ for all $j\in\{1,2,\dots,i\}$. Note that the claim holds for $j=i$:  sumdel$_S(i)=m^{i-i}=1$, and there is indeed only one copy of cbin$_i(i)$ in $S$ by the second condition of Definition \ref{def:pat}. Now assume that the claim holds for some $j\in\{2,\dots,i\}$, so that $S$ contains 
precisely $m^{i-j}$ delimiters of the form cbin$_i(t)$,  where $t\geq j$. Consider the subsequence of $S$ which is composed of these ``highly valued'' delimiters. By the third condition of Definition \ref{def:pat}, $S$ contains $m-1$ occurrences of cbin$_i(j-1)$ both to the left of the first element 
and also between any two successive elements of this subsequence. To calculate sumdel$_S(j-1)$, we add the number of all these occurrences of  cbin$_i(j-1)$ to sumdel$_S(j)$, obtaining $(m-1)m^{i-j}+m^{i-j}=m^{i-(j-1)}$, and proving the claim.
Since sumdel$_S(1)$ equals the total number of delimiters in $S$, we see that 
$|S|=m^{i-1}$.

It remains to show that there exists a unique delimiter sequence $S_{i,m}$ that is  well-ordered for $(i,m)$ for every $i>1$ and $m>m_i$. We describe how to use Definition \ref{def:pat} to construct $S_{i,m}$, by assigning delimiters of the form cbin$_i(j)$ to specified ``slots'' of $S$ stage by stage for all $j$ from $i$ down to 1. In the first stage, we assign  cbin$_i(i)$ to the last (i.e., $m^{i-1}$th) slot. At this point, we have a ``vacant interval'' of $m^{i-1}-1$ empty slots to the left of the single occurrence of the delimiter cbin$_i(i)$. The assignments to be made to some of those slots in the later stages will break this interval into shorter and shorter vacant intervals, until all remaining slots are filled in the final stage with the assignment of the cbin$_i(1)$'s. We note that Definition \ref{def:pat} dictates  each stage to assign its delimiters  so that all resulting vacant intervals have equal length, since the procedure that will fill those intervals (and therefore the total number of delimiters that will fill any such interval completely) is identical for every interval. This forces an ordering where, for all $j>2$ and $l\geq j$, every pair of successive occurrences of delimiters of the form cbin$_i(l)$ in the assignment are equidistant (by, say, $d_j$ slots) from each other in the resulting sequence. Furthermore, the leftmost occurrence of such a delimiter must also be $d_j$ slots away from the beginning of the sequence.\footnote{In more detail,  every such  cbin$_i(l)$ delimiter must be located at a slot whose index  in the sequence is a multiple of $m^{l-1}$.}  Since each stage is constrained to assign the associated delimiters to unambiguously determined locations, the resulting sequence is unique.
\end{proof}
It will be shown shortly that adherence to this ordering pattern can be checked in polynomial time by a 2QCFA. 

For every $i$, let 
\[PAL_i=\Set{p \mid p\in \textit{PAL}, \left| p \right|=m^i, m\geq m_i},\]
and let punc$_i(p)$ denote the string that would be obtained from a member $p$ of $PAL_i$ by punctuating it with binary delimiters using the procedure described above. Formally, 
\[PPAL_i = \Set{w \mid w=\text{punc}_i(p), p \in PAL_i}.\]

\paragraph{The padding procedure.} Let the function segl($w$) denote the length of the first segment (i.e., the prefix up to the leftmost occurrence of a member of $\{0,1\}$) for any given member $w$ of $PPAL_i$.
As described earlier, any such string is associated with a number $m$, which is the common length of the $m^{i-1}$ segments that constitute its underlying palindrome. We note that $m=\text{segl}(w)$. 
The following definitions encapsulate a padding procedure that is used to convert members of $PPAL_i$ (for any $i$) to much longer strings.

\[
p_{i,l}=
\begin{cases}
    \$\mathtt{a}0^{c_i}\mathtt{a}, & \text{if } l = 1, \\
    \$\mathtt{a}^{l}(0^{c_i}\mathtt{a}^l)^{2^l-1}p_{i,l-1}, & \text{if } l > 1.
\end{cases}
\]




\[
PPPAL_i=\{w(\mathtt{a}^m0^{c_i})^{2^m-m^{i-1}-1}\mathtt{a}^m p_{i,m-1} \mid w \in PPAL_i, m=\text{segl}(w)\}
\]

The  definitions presented above correspond to extending some string $w\in PPAL_i$ stage by stage, concatenating a new substring to the right at each stage, as follows:
\begin{itemize}
    \item Append the substring $(\mathtt{a}^m0^{c_i})^{2^m-m^{i-1}-1}\mathtt{a}^m$
    \item[] \{After this stage, the extended string consists of $2^m$ segments of length $m$. Of these, the first $m^{i-1}$ are of the form $\Sigma_{\mathtt{ab}}^m$, and the remaining $2^m-m^{i-1}$ are of the form $\mathtt{a}^m$. Each neighboring pair of  newly added segments are separated by the delimiter substring $0^{c_i}$.\}
    \item For each integer $l$ from $m-1$ down to and including 1, append the substring $\$\mathtt{a}^{l}(0^{c_i}\mathtt{a}^l)^{2^l-1}$
    \item[] \{Each iteration of this loop appends a (shorter) \textit{block} consisting of $2^l$ $\mathtt{a}$-segments of length $l$. Neighboring $\mathtt{a}$-segments of equal length are separated by the delimiter substring $0^{c_i}$, whereas neighboring $\mathtt{a}$-segments of different lengths are separated by the symbol $\$$ throughout the resulting string.\}
\end{itemize}

For any $i>0$, the corresponding \textit{padded punctuated palindromes language} $PPPAL_i$ is the set of strings obtained by applying the padding procedure described above to every member of the corresponding punctuated palindromes language $PPAL_i$.

Let us  examine the common pattern 
imposed by our padding procedure on the resulting strings.  Employing the terminology introduced above, we see that each string $w$ in $PPPAL_i$ consists of $m$ blocks, where $m=\text{segl}(w)$. Neighboring blocks are separated by $\$$ symbols. Within each block, there are multiple segments of 
equal length, separated by binary delimiter substrings. The rightmost block contains two segments of length 1. Every other block contains twice as many segments as the block  to its right, and the length of its segments  is one more than the corresponding length in the  block to its right.

\begin{lemma}\label{logn}
  For any $i,n>0$, if $m$ is the length of the leftmost segment of a string $s\in PPPAL_i$ such that $|s|=n$, then $m<\log n$, and there exists a positive constant $d_i$  such that $m > d_i \log n$.  
\end{lemma}
\begin{proof}
It is easy to prove by induction that any $m$-block string in $PPPAL_i$  contains $2^{m+1}-2$ segments. The total number of ``separators'' between segments is therefore $2^{m+1}-3$. Of these, $m-1$ are the $\$$ symbols which separate neighboring blocks, so there are $2^{m+1}-m-2$ binary delimiter substrings. For any fixed $i$, the total number of symbols occurring in the separators is therefore $c_i(2^{m+1}-m-2)+m-1$. It remains to consider the combined length of the segments. This is  $\sum_{l=1}^{m} l 2^l$, which is seen by another simple induction to equal
$(m-1)2^{m+1}+ 2$. So for any string of length $n$ matching  this  pattern, and having  $m$ symbols in its leftmost segment, 
\begin{equation}
    n=(m-1)2^{m+1}+c_i(2^{m+1}-m-2)+m+1. \label{eq:mlogn}
\end{equation}
This lets one conclude that  
$m<\log n$, and that there exists a positive constant $d_i$  such that $m > d_i \log n$.
\end{proof}

Note that the  analysis in the proof of Lemma \ref{logn} uses only the numbers, lengths and ordering of the segments and separators in the extended string, and is valid for any string matching this general pattern, even if it has no prefix that is a member of some $PPAL_i$. Just as in Theorem \ref{thm:mrpal}, the following result  will be based on the fact that the dominant term of the expected runtime of an optimal small-space QTM recognizing $PPPAL_i$ is the one associated with the task of determining whether the relevant prefix of the input contains a palindrome, and a 2QCFA can quickly check all other syntactic requirements associated with membership of $PPPAL_i$  with a very small probability of error.

\begin{theorem}\label{thm:pppal}
        For all $i\geq1$, $PPPAL_i \in \mathsf{BQTISP}(2^{O((\log n)^i)},O(1))$.
\end{theorem}
\begin{proof}

Figure \ref{fig:qalg2} is a template for constructing, for any desired $i\geq 1$, a 2QCFA $M_{PPPAL_i}$ that recognizes the language $PPPAL_i$ with negative one-sided bounded error. 
In the figure, $BD_i$ denotes the set of allowed binary delimiter substrings, i.e. all strings of the form cbin$_i(l)$ for $l\in\{1,\dots,i\}$, as well as $0^{c_i}$. 
The general structure of the algorithm is similar to that of Figure \ref{fig:qalg1}: The main loop (M) is exited successfully with high probability only if the input matches the ``easily checkable'' pattern described above. The final stage (P) runs only on a prefix of the tape whose length is polylogarithmic ($\Theta((\log n)^i)$) in the overall input length $n$ with high probability, ensuring that the expected runtime of $M_{PPPAL_i}$ is in $2^{O((\log n)^i)}$. The parameters $\varepsilon$ and $k_{\varepsilon}$ play the same role as their counterparts in the algorithm of Figure \ref{fig:qalg1}. A more detailed description follows.

\begin{figure}[htb!]
    \caption{$M_{PPPAL_i}$}
    \label{fig:qalg2} 
    \begin{turing}[(RW)]{V}{\textit{  }}
        \titem[(R)]{Check whether the input begins with a symbol  in $\Sigma_{\mathtt{ab}}$, and ends with the postfix  $\$\mathtt{a}0^{c_i}\mathtt{a}$. If not, \reject.}
        \titem{Check whether 
        each separator (i.e. any substring sandwiched between members of $\Sigma_{\mathtt{ab}}$) in the input is either the symbol $\$$, or a member of $BD_i$. If not, \reject.}
        \titem{Check whether the separator 
        cbin$_i(i)$ appears exactly once in the string. If not, \reject.}
        \titem{Check whether  any separator of the form $0^{c_i}$ or $\$$ exists to the left of the separator cbin$_i(i)$. If so, \reject.}
        \titem{Check whether any symbol $\mathtt{b}$ or any
         separator other than   $0^{c_i}$ and $\$$ exists to the right of the separator cbin$_i(i)$. 
         If so, \reject.}
        \titem[(M)]{
            Repeat ad infinitum:}
        \ttitem{
            \{
            Let $p_r$ denote the  position of the 
            right endmarker.\}}
        \ttitem{
            While the input contains at least one separator located to the left of $p_r$, do the following:}
        \tttitem{
            \{Let $p_m$ denote the position of the rightmost symbol of the nearest separator, to be named $Sep$,  to the left of $p_r$.\}}
        \tttitem{\{Let $p_l$ denote the position of the rightmost symbol of the nearest separator  (or the left endmarker,
if no such separator exists)   to the left of $p_m$.\}}
        \tttitem[(C1)]{    
            If $Sep\in BD_i$, run \textsc{SameLength}$(\langle\Sigma _{\mathtt{ab}},(p_l,p_m)\rangle,\langle\Sigma_{\mathtt{ab}},(p_m,p_r)\rangle,\varepsilon)$.}
        \tttitem[(C2)]{    
            If $Sep=\$$, run \textsc{SameLength}$(\langle\{\mathtt{a}\},(p_l,p_m)\rangle,\langle\{\mathtt{a},\$\},[p_m,p_r)\rangle,\varepsilon)$.}
        \tttitem{
            If the symbol at position   $p_r$ is 
            $\$$ or the right endmarker, then check if the input contains another $\$$ to the left of $p_r$. If so, do the following:} 
        \ttttitem{
            \{Let $p_m$ denote the position of the nearest $\$$ located to the left of $p_r$.\}}
        \ttttitem{\{Let $p_l$ denote the position of the nearest $\$$  (or the left endmarker,
if no such $\$$ exists) located to the left of $p_m$.\}} 
        \ttttitem[(TW)]{    
            Run
            \textsc{TwiceAsLong}$(p_l,p_m,p_r,\varepsilon)$.}
        \tttitem{    
            \{Update $p_r$ to the position of the leftmost symbol of the nearest separator to the left of its present value.\}
            }
        \ttitem{If $i>1$, do the following:}
        \tttitem[(L)]{For each integer $j$ from $i$ down to and including 2, do the following:}
        \ttttitem{\{Let $p_m$ be the position of the $\mathtt{a}$ symbol immediately to the right of the  delimiter  cbin$_i(i)$.\}}
        \ttttitem{While  $p_m$  does not equal the position of the left endmarker, do the following:}
        \tttttitem{\{Let  $p_r$ denote the position of the leftmost symbol of  the nearest  separator located to the right of  $p_m$.\}}
        \tttttitem{\{Let  $p_l$   denote the position of the rightmost symbol of the nearest  separator of the form cbin$(k)$, where $k\geq j$ (or the left endmarker, if no such separator exists) located to the left of the binary delimiter which is immediately to the left of  $p_m$.\}}
        \tttttitem[(C3)]{Run \textsc{SameLength}$(\langle
\{\text{cbin}_i(j-1)\},
(p_l,p_m)\rangle,\langle
\Sigma_{\mathtt{ab}},
(p_m,p_r)\rangle,\varepsilon)$.}
        \tttttitem{\{If $p_l$ is the position of the left endmarker, update $p_m$ to that position as well. Otherwise, update $p_m$ to the position immediately to the right of $p_l$.\}}
        \ttitem[(RW)]{Move the tape head to the leftmost input symbol.}
        \ttitem{While the currently scanned symbol is not an endmarker, do the following:}
        \tttitem{Simulate a classical coin flip. If the result is ``heads'', move  one cell to the right. Otherwise, move  one cell left.}
        \ttitem{If the random walk ended at the right endmarker, do the following:}
        \tttitem{Simulate $k_{\varepsilon}$ coin flips. If all results are ``heads'', exit loop (M).}
        \titem[(P)]{
        Move the input head to the left endmarker and run \textsc{Pal}$($cbin$_i(i),\varepsilon)$.}
    \end{turing}
\end{figure}

As its counterpart $M_{RPAL_i}$ in the proof of Theorem \ref{thm:mrpal}, $M_{PPPAL_i}$ starts with a deterministic format check in stage (R). This control, which can  be implemented in a single non-stop left-to-right scan of the tape, is passed without rejection if
\begin{itemize}
    \item the input consists of multiple $\$$-separated blocks, each of which contain multiple segments separated by binary delimiter substrings, and 
    \item the input is of the form $x$cbin$_i(i)y$, where cbin$_i(i)$ is in the leftmost block, neither the prefix $x$ nor the postfix $y$ contain cbin$_i(i)$, $x$ consists  of one or more segments of the form $\Sigma_{\mathtt{ab}}^+$ separated by delimiters of the form cbin$_i(j)$ where $1\leq j<i$, and all segments in $y$ are of the form $\mathtt{a}^+$, and are separated by separators of the form $0^{c_i}$ or   $\$$. 
\end{itemize}

In each iteration of the main loop (M), $M_{PPPAL_i}$ goes through all of the following constraints imposed by the padding and punctuation procedures on members of $PPPAL_i$:
\begin{enumerate}
    \item Any two segments separated by a binary delimiter substring should have the same length, and the \textsc{SameLength} procedure introduced in the proof of Theorem \ref{thm:mrpal} (Section \ref{sec:hier1}) is used to check this condition in stage (C1).
    \item Any two segments separated by a $\$$ belong in different blocks, so the segment on the right should be one symbol longer than the one on the right. The \textsc{SameLength} procedure is employed to check this condition by comparing the number of symbol occurrences in the left segment with the length of a sequence defined to consist of the $\$$ separator \textit{and} all the symbols of the right segment in stage (C2).
    \item The condition that each block should contain twice as many segments as the one to its right is checked by a new procedure named \textsc{TwiceAsLong}, which can be obtained by a straightforward modification of \textsc{SameLength}: To verify that a subsequence $\mathbf{L}$ on the left is exactly twice as long as another subsequence $\mathbf{R}$ on the right, one can first check whether $\mathbf{L}$ has an even number of elements (rejecting otherwise), and then perform a length comparison based on \textsc{SameLength} where only half of the members (the ones whose indices are even numbers, starting counting from 1 and from left to right) of $\mathbf{L}$ and all members of $\mathbf{R}$ are reflected on the ``virtual tape'' presented to the core comparison algorithm $M_{EQ}$. In stage (TW), \textsc{TwiceAsLong} compares the sequence of separators that precede the selected segments of the block on the left (in the interval $(p_l,p_m))$ with the sequence of separators that follow the segments of the block on the right (in the interval $[p_m,p_r)$) in this manner, and terminates with rejection unless it arrives at the conclusion that  the condition is satisfied. 
    The maximum probability $\varepsilon$ with which \textsc{TwiceAsLong} will fail to reject a subsequence pair violating the condition can be tuned as described in Section \ref{sec:hier1}.
    \item Finally, the inner loop (L) checks whether the sequence of delimiters in the prefix of the input terminated by  cbin$_i(i)$ is well-ordered according to Definition \ref{def:pat}. Recall that, for each $j>1$, Definition \ref{def:pat} requires the delimiter cbin$_i(j-1)$ to occur exactly $m-1$ times in several carefully specified subregions of the tape, where $m$ is the common length of all segments of the member of $PPAL_i$ under consideration. In stage (C3), $M_{PPPAL_i}$ checks each such condition by using \textsc{SameLength} to compare the number of occurrences of cbin$_i(j-1)$ in the relevant subregion (in the interval $(p_l,p_m)$) with  a sequence in  the interval $(p_m,p_r)$, whose length is  $l-1$, where $l$ is the length of the segment immediately to the right of that subregion in the input. Note that  $M_{PPPAL_i}$ can 
    reach loop (L) without rejection with probability at most $\varepsilon$ if all segments which are used as ``rulers'' in this stage do not have the same length $m$ as the very first segment.
\end{enumerate}
    
As with the machine $M_{RPAL_i}$ of Theorem  \ref{thm:mrpal}, the probability that $M_{PPPAL_i}$ rejects its input string $s$  of length $n$ in a single iteration of the main loop (M) is 0 if $s\in PPPAL_i$, and at least $1-\varepsilon$ if $s$ does not match the punctuation and padding patterns associated with members of $PPPAL_i$, and we again use the random walk in stage (RW) to ensure that (M)'s expected number of iterations is $2^{k_{\varepsilon}}(n+1)$ (unless the machine rejects earlier, having detected a violation of the above-mentioned patterns). 
 
 Each iteration of (M) consists of $O(n)$ length comparisons, each of which run in polynomial time. The runtime of the random walk is linear. The total expected runtime of (M) is therefore polynomially bounded. 

Stage (P) of $M_{PPPAL_i}$ checks whether the sequence of symbols in $\Sigma_{\mathtt{ab}}$ contained in the input prefix that ends with cbin$_i(i)$ is a palindrome or not, with a one-sided error bound of $\varepsilon$. As in the proof of Theorem \ref{thm:mrpal}, let $c_{\varepsilon}$ be a positive constant so that the runtime of the \textsc{Pal} procedure would not exceed $2^{c_{\varepsilon}n}$ if it operated on the entire input of length $n$ with this error bound, and set  $k_{\varepsilon}$ to the smallest positive integer that makes the product $\varepsilon^{2^{k_{\varepsilon}}(n+1)}2^{c_{\varepsilon}n}$ near 0 as $n$ grows. Following the argument in that proof, this enables us to conclude that the  runtime of $M_{PPPAL_i}$ is determined by the asymptotic behavior of the \textsc{Pal} procedure on strings that match the padding and punctuation patterns controlled by the main loop (M).


In this high-probability case, the string in $\Sigma_{ab}^+$ that the \textsc{Pal} procedure checks in stage (P) for palindromeness will have length $m^i$, where  $m$ denotes the length of  the leftmost segment of the input of length $n$.
By Lemma \ref{logn}, $m<\log n$. 
 We conclude that the expected runtime of both that stage and the overall algorithm $M_{PPPAL_i}$ is $2^{O((\log n)^i)}$.
\end{proof}

Note that the machine $M_{PPPAL_1}$ runs in polynomial time. This allows us to prove an enhanced quantum advantage result, establishing that polynomial-time 2QCFA's can perform tasks which are impossible for classical small-space machines, even when there is no restriction on the runtime of the classical machine. For this purpose, we will be using the following theorem of Freivalds and Karpinski, which builds on previous work by Dwork and Stockmeyer \cite{DS92}:

\begin{theorem}\label{thm:FK}
\cite{FK94} Let $A, B \subseteq \Sigma^*$ with $A \cap B = \emptyset$. Suppose there is an infinite set $I$ of positive integers and functions $g(m), h(m)$ such that $g(m)$ is a fixed polynomial in $m$, and for each $m \in I$, there is a set $W_m$ of words in $\Sigma^*$ such that:
\begin{enumerate}
    \item $\abs{w} \leq g(m)$ for all $w \in W_m$,
    \item there is a constant $c > 1$ 
    such that $\abs{W_m} \geq c^m$, 
    \item for  every $w, w' \in W_m$ with $w \neq w'$, there are words $u, v \in \Sigma^*$ such that:
    \begin{enumerate}
        \item $\abs{uwv} \leq h(m), \abs{uw'v} \leq h(m)$, and
        \item either
            $uwv \in A$ and $uw'v \in B$, or 
            $uwv \in B$ and $uw'v \in A$.
    \end{enumerate}
\end{enumerate}
Then, if a probabilistic Turing machine with space bound $S(n)$ separates $A$ and $B$, then $S(h(m))$ cannot be in $\oh{\log m}$.
\end{theorem}

\begin{theorem}\label{thm:adv}
    $\mathsf{BQTISP}(n^{O(1)},O(1))\not\subseteq \mathsf{BPSPACE}(o(\log \log n))$.
\end{theorem}
\begin{proof}
    By Theorem \ref{thm:pppal}, $PPPAL_1\in \mathsf{BQTISP}(n^{O(1)},O(1))$. We will use the following argument adapted from \cite{SG24} to show that this language cannot be recognized by any probabilistic Turing machine using $o(\log \log n)$ space.
    
    Note that the task of recognizing $PPPAL_1$ is identical to the task of separating $PPPAL_1$ and $\setneg{PPPAL_1}$. 

In the template of \Cref{thm:FK}, let $I$ be a set containing all positive even numbers, $g(m)=m$,  
and $W_m=\Sigma_{\mathtt{ab}}^{\sfrac{m}{2}}$. Note that Condition (2) of the theorem is satisfied by $c=\sqrt{2}$. 
Let $h$ be the function that maps each number $m$ to the length $n$ of any string in $PPPAL_1$ whose leftmost segment has length $m$, i.e. $h(m)=m2^{m+1}-1$ by setting $c_i=1$ in Equation  \ref{eq:mlogn}.

For every $w, w' \in W_m$, let
the corresponding $u$ be the empty string, and the corresponding $v$ be the string $w^R  1 t$, where $t$ is the unique postfix that satisfies $ww^R  1 t\in PPPAL_1$.  Clearly, $uwv \in PPPAL_1$ and $uw'v \notin PPPAL_1$. 

By Lemma \ref{logn}, we have $m\in\Theta(\log n)$, which enables us to use Theorem \ref{thm:FK} to conclude that no probabilistic Turing machine with space bound in $o(\log \log n)$ can recognize $PPPAL_1$.
\end{proof}


We now  state our second hierarchy theorem.

\begin{theorem}\label{thm:hier2}
    For each $i \geq 1$, $\mathsf{BQTISP}(2^{O((\log n)^i)},o(\log \log n))\subsetneq \mathsf{BQTISP}(2^{O((\log n)^{i+1})},o(\log \log n))$.
\end{theorem}

\begin{proof}
    For each $i$,  $PPPAL_{i+1}\in \mathsf{BQTISP}(2^{O((\log n)^{i+1})},O(1))$ by Theorem \ref{thm:pppal}. Let us also show that 
    \[PPPAL_{i+1} \notin \mathsf{BQTISP}(2^{O((\log n)^i)},o(\log \log n)).\]

For every sufficiently large even number $m$, let $PL_{i+1,m}$ be the set of all strings in $PPPAL_{i+1}$ whose leftmost segments are of length $m$. The punctuation and padding procedures dictate that   all strings in $PL_{i+1,m}$ in have the same length, say, $n$, and contain exactly $m^i$ segments of the form $\Sigma_{\mathtt{ab}}^m$ to the left of the unique occurrence of the substring cbin$_{i+1}(i+1)$.  We know by Lemma \ref{logn} that there exists a constant 
$d_{i+1}>0$  such that $m > d_{i+1} \log n$.

 For every string $s$ in $PL_{i+1,m}$, let $l_s$ denote the minimal prefix of $s$ that contains the leftmost  $m^i/2$ segments of the form $\Sigma_{ab}^m$, that is, $l_s$ contains the left half of the $m^{i+1}$-symbol palindrome underlying $s$. For every $s$, define $r_s$ as the unique postfix that satisfies $l_s r_s=s$. 

 
 Let $W_n$ be the set $\{w \mid w=l_s, s\in PL_{i+1,m}\}$.
 
 Since every distinct pair of strings $l_s$, $l_{s'}$ in $W_n$ are distinguished by the string $r_s$,  we conclude that 
$D_{PPPAL_{i+1}}(n)\geq |W_n|\geq 2^{\frac{m^{i+1}}{2}}   \geq 2^{\frac{(d_{i+1} \log n)^{i+1}}{2}}    \geq 2^{\frac{d_{i+1}^{i+1}}{2}(\log n)^{i+1}}$
for infinitely many values of $n$.

Let us assume that 
    \[PPPAL_{i+1} \in \mathsf{BQTISP}(2^{O((\log n)^i)},S(n)),\]
for some $S(n)\in o(\log \log n)$.

Since $D_{PPPAL_{i+1}}(n)\geq  2^{\frac{d_{i+1}^{i+1}}{2}(\log n)^{i+1}}$, we know that 
\[\log \log D_{PPPAL_{i+1}}(n)\geq   \log \frac{d_{i+1}^{i+1}}{2} + (i+1)\log \log n,\] 
meaning that $\log \log D_{PPPAL_{i+1}}(n)\in   \Omega(\log \log n)$, and $S(n)$ is therefore small enough for  Theorem \ref{thm:remscrimQTMtradeoff} to be applicable.

Using that theorem, we see that there must exist \(e_0,  b_0 >0 \) such that
\[
2^{e_0(\log n)^i} \in \Omega\big(2^{\frac{d_{i+1}^{i+1}}{2}(\log n)^{i+1} 2^{-b_0 S(n)}-b_0 S(n)}\big).
\]

Since $S(n)\in o(\log \log n)$, for any positive constant $c_0$, $S(n)\leq \frac{c_0}{b_0} \log \log n$ for all sufficiently large $n$. We see that, for sufficiently large $n$,
\[2^{-b_0 S(n)}\geq 2^{-b_0 \frac{c_0}{b_0} \log \log n}=2^{-c_0  \log \log n}=(\log n)^{-c_0},\]
yielding
\[2^{e_0(\log n)^i} \in \Omega\big(2^{\frac{d_{i+1}^{i+1}}{2}(\log n)^{i+1-c_0}-b_0 S(n)}\big).\]
This is supposed to be true for all positive $c_0$. Let us consider  setting $c_0=1/2$.
Since any positive power of $\log n$ eventually dominates 
$\log \log n$, 
\[b_0 S(n) \leq \frac{\frac{d_{i+1}^{i+1}}{2}}{2}(\log n)^{i+1-1/2}\]
for all large $n$, which yields
\[2^{e_0(\log n)^i} \in \Omega\big(2^{\frac{{d_{i+1}^{i+1}}}{4} (\log n)^{i+1/2}}\big).\]
This is a contradiction, since $(\log n)^i \in o((\log n)^{i+\frac{1}{2}})$. We conclude that our assumption that $PPPAL_{i+1} \in \mathsf{BQTISP}(2^{O((\log n)^i)},S(n))$ is false, completing the proof.
\end{proof}

\section{Concluding remarks}\label{sec:conc}

Remscrim \cite{R21} notes the following time-space tradeoff: No QTM that uses $o(\log n)$ space can recognize \textit{PAL} in polynomial time, whereas the language can be recognized by a polynomial-time \textit{deterministic} Turing machine  that uses $O(\log n)$ space. A similar result can be proven for another language associated with the space threshold $\Theta(\log \log n)$, as we argue below.

For any integer $l\geq 0$, let bin$(l)$ denote the shortest string representing $l$ in the binary notation. Stearns et al. \cite{SHL65} describe a deterministic TM $M_{SHL}$ that recognizes the  language 
\[
\textit{SHL}=\{\text{bin}(0)\$\text{bin}(1)\$\text{bin}(2)\$\cdots\$\text{bin}(k) \mid k>0\}
\]
in $\Theta(\log \log n)$ space. It is easy to show that the rightmost ``segment'' (i.e. the postfix following the rightmost $\$$) of any $n$-symbol member of \textit{SHL} is of length $m\in\Theta(\log n)$, by using an argument similar to the one in Lemma \ref{logn}. $M_{SHL}$ uses its work tape to repeatedly count from 1 to at most $m$.

Consider changing the input alphabet by  adding a new delimiter symbol $\#$, and replacing the symbols $0$ and $1$ by  the four ``multi-track'' symbols $(0,\mathtt{a})$, $(0,\mathtt{b})$, $(1,\mathtt{a})$, and $(1,\mathtt{b})$. One can now visualize the input tape as containing binary numbers (as required in the definition of \textit{SHL}) on the top track, and strings on the alphabet $\Sigma_{\mathtt{ab}}$ on the bottom track. We define a new language, \textit{LLL}, on this new alphabet as follows: The top tracks of the members of $LLL$ contain the numbers from 0 to some integer $k$, exactly as in \textit{SHL}. The segments of members of \textit{LLL} are again delimited by $\$$'s, with the following exception: Let $m$ be the length of the rightmost segment. Counting from the right, the delimiter between the $m$th segment and the $(m+1)$st segment of a member of \textit{LLL} is not $\$$, but $\#$. Finally, the string sitting on the bottom tracks of the symbols in the postfix following the $\#$ is a palindrome. \textit{LLL} is the set of all strings satisfying all these conditions.

Note the similarities between \textit{LLL} and $PPPAL_2$, the language in the second level of the hierarchy examined in Section \ref{sec:hier2}: In both cases, recognizing the language involves determining whether a substring of length $\Theta((\log n)^2)$ is a palindrome or not. The argument used in the proof of Theorem \ref{thm:hier2} to show that no $o(\log \log n)$-space QTM can recognize $PPPAL_2$ in $2^{O(\log n)}$ (i.e., polynomial) time is therefore applicable to \textit{LLL} as well. On the other hand, it is easy for a deterministic polynomial-time TM  to recognize \textit{LLL} using $O(\log \log n)$ space, by employing the method of $M_{SHL}$ to check the top track, using a counter that counts up to $m$ to verify that the $\#$ delimiter is placed correctly, and performing the palindromeness check by comparing symbols at equal distances from the left and right ends of the region of interest by utilizing pointers of size at most $\log ((\log n)^2)\in O(\log \log n)$.

We conclude with an open question. In Theorem \ref{thm:adv}, we used $PPPAL_1$ as the first example of a language that is simultaneously recognizable by a polynomial-time 2QCFA and unrecognizable (regardless of time bounds) by any small-space classical TM. Remscrim \cite{R20} has shown that the word problems of all finitely generated virtually
 abelian groups are in $\mathsf{BQTISP}(n^{O(1)},O(1))$. It would be interesting to determine whether any such language can be shown to be outside $\mathsf{BPSPACE}(o(\log \log n))$ as well. In particular, is the word problem of the fundamental group of the Klein bottle in $\mathsf{BPSPACE}(o(\log \log n))$?

\section*{Acknowledgments}

The author thanks Utkan Gezer for all the useful discussions and for the preparation of the LaTeX infrastructure used in the figures. 

\bibliographystyle{abbrvnat}

\bibliography{references} 

\begin{appendices}
\section*{Appendix}
\label{sec:app}

\section{Running a 2QCFA on a virtual tape}%
\label{sec:subro}

As an example of how a 2QCFA that performs a certain task is adapted to obtain an algorithm that can be used as a subroutine for performing a generalized version of that task, we describe how the machine $M_{EQ}$ (Section \ref{sec:defs}) can form the basis for the construction a submachine embodying the \textsc{SameLength} procedure used in Sections \ref{sec:hier1} and \ref{sec:hier2} for a given tuple of parameters.

Recall that $M_{EQ}$ recognizes the language 
    $\textit{EQ}=\Set{\mathtt{a}^n \mathtt{b}^n | n\geq 0}$ with negative one-sided bounded error. The allowed error bound $\varepsilon$ can be set to any desired positive value by arranging the core algorithm presented in \cite{R20} accordingly. Therefore, strictly speaking, there are infinitely many versions of $M_{EQ}$,  corresponding to different values of $\varepsilon$. 
In the algorithm templates in Figures \ref{fig:qalg1} and \ref{fig:qalg2}, $\varepsilon$ appears as a  parameter, and it is to be understood that the smallest $M_{EQ}$ version (i.e. the one with the fewest states in its classical state set) consistent with the specified value of $\varepsilon$ will be used. (If more than one version satisfies this condition, the choice is arbitrary.)

The following arguments about the \textsc{SameLength} procedure are based on the assumption that
there exists a unique way of writing the string $\lend w \rend$ sitting on the input tape as a concatenation of \textit{tokens}, i.e., strings of length at least 1 on the alphabet $\Sigma\cup \{\lend,\rend\}$, from  a given finite set of interest. (The deterministic checks at the stages named (R) ensure that this promise will be true for the associated token sets if the machine reaches the later stages in the algorithm templates in Figures \ref{fig:qalg1} and \ref{fig:qalg2}.) 

While $M_{EQ}$'s input head initially scans the left endmarker  and can reach all cells up to and including the one containing the right endmarker as dictated by the algorithm, \textsc{SameLength} can be set to operate only within a specific subregion of the tape. As mentioned in Section \ref{sec:hier1}, this subregion is further subdivided into two intervals  named 
$I_{\mathbf{L}}$
and $I_{\mathbf{R}}$, and the machine is required to compare the length of a subsequence of tokens from a specified set in $I_{\mathbf{L}}$ with the length of another subsequence in 
$I_{\mathbf{R}}$. In Figures \ref{fig:qalg1} and \ref{fig:qalg2}, the boundaries of these intervals are determined by the three signposts $p_l$, $p_m$ and $p_r$, which are specified before every activation of the \textsc{SameLength} procedure. We note that all these signposts are ``locatable'', i.e., for each signpost of any such instance of \textsc{SameLength}, there exists a 2DFA subroutine that moves the head from its position immediately before the activation of that \textsc{SameLength} procedure to that signpost, and halts there, for any input which has passed the checks at stage (R).
 The constructed submachine for \textsc{SameLength}  will therefore be able to determine the position of its head relative to the intervals $I_{\mathbf{L}}$ and $I_{\mathbf{R}}$ at every point during its execution.

To complete the specification of a particular \textsc{SameLength} task, one also specifies two sets of tokens, $S_{\mathbf{L}}$ and $S_{\mathbf{R}}$. The procedure will compare the number of members of $S_{\mathbf{L}}$ in the interval $I_{\mathbf{L}}$ with the number of members of $S_{\mathbf{R}}$ in $I_{\mathbf{R}}$.

 We describe the submachine $M_{SL}$ embodying \textsc{SameLength}$(\langle S_{\mathbf{L}},I_{\mathbf{L}}\rangle,\langle S_{\mathbf{R}},I_{\mathbf{R}}\rangle,\varepsilon)$.  
 $M_{SL}$ starts execution after locating the head at position $p_l$.
 $M_{SL}$  simulates $M_{EQ}$ on a virtual tape that it obtains by ``translating'' the actual tape content: If the head is at position $p_l$ (resp. $p_r$), $M_{SL}$ will ``feed''  the left (resp. right) endmarker to $M_{EQ}$. If the head is in the interval $I_{\mathbf{L}}$ (resp. $I_{\mathbf{R}}$), $M_{SL}$ will feed the symbol $\mathtt{a}$ (resp. $\mathtt{b}$) to  $M_{EQ}$ whenever it scans a token from   $S_{\mathbf{L}}$ (resp. $S_{\mathbf{R}}$). $M_{SL}$ concludes its execution when the simulation reaches a halting state of $M_{EQ}$, rejecting if $M_{EQ}$ rejects.

As explained in Section \ref{sec:defs}, each step of a 2QCFA like $M_{EQ}$ consists of a quantum transition determined by the present state $q$ and the currently scanned input symbol $\sigma$, followed by a classical transition (updating the state and the head position) determined by $q$, $\sigma$, and the  outcome  of the quantum transition. We note that $M_{SL}$ is not guaranteed to be able to simulate each computational step of $M_{EQ}$ in a single step of its own, since  it may require multiple steps to move the head to skip over tokens  which do not belong to the set of interest for the currently scanned interval in the direction dictated with the simulated classical transition of $M_{EQ}$, and to traverse tokens within the set of interest whose lengths are longer than $1$. The expected runtime of $M_{SL}$ is polynomially bounded, as it cannot exceed the runtime of $M_{EQ}$, which is  known \cite{R20} to be polynomially bounded, on an input as long as the whole input string.

\end{appendices}
\end{document}